\documentclass[11pt]{article}

\usepackage[margin=1in]{geometry}
\usepackage{graphicx}
\usepackage[pdftex,colorlinks=true,linkcolor=blue,citecolor=blue,urlcolor=black]{hyperref}
\usepackage{amsmath, amsthm, amssymb}
\usepackage{subfigure}
\usepackage{comment}
\usepackage{url}
\usepackage{xcolor}

\newcommand{\R}{\mathbb{R}}

\newcommand{\Z}{\mathbb{Z}}

\newcommand{\E}{\mathbb{E}}
\newcommand{\F}{\mathbb{F}}

\newcommand{\ket}[1]{| #1 \rangle}
\newcommand{\bra}[1]{\langle #1|}
\newcommand{\ip}[2]{\langle #1|#2 \rangle}
\newcommand{\bracket}[3]{\langle #1|#2|#3 \rangle}

\newcommand{\be}{\begin{equation}}
\newcommand{\ee}{\end{equation}}
\newcommand{\bea}{\begin{eqnarray}}
\newcommand{\eea}{\end{eqnarray}}
\newcommand{\bes}{\begin{equation*}}
\newcommand{\ees}{\end{equation*}}
\newcommand{\beas}{\begin{eqnarray*}}
\newcommand{\eeas}{\end{eqnarray*}}

\newtheorem{thm}{Theorem}
\newtheorem*{thm*}{Theorem}
\newtheorem{cor}[thm]{Corollary}

\newtheorem{lem}[thm]{Lemma}
\newtheorem{fact}[thm]{Fact}
\newtheorem*{lem*}{Lemma}
\newtheorem{prop}[thm]{Proposition}

\theoremstyle{definition}
\newtheorem{dfn}{Definition}

\newcommand{\boxdfn}[2]{
\begin{figure}[h]
\begin{center}
\noindent \framebox{
\begin{minipage}{11cm}
\begin{dfn}[#1]
\ \\ \\
#2
\end{dfn}
\end{minipage}
}
\end{center}
\end{figure}
}

\begin{document}

\title{On exact quantum query complexity}

\author{Ashley Montanaro\footnote{{\tt am994@cam.ac.uk}}, Richard Jozsa and Graeme Mitchison\\[11pt] {\small Centre for Quantum Information and Foundations, DAMTP, University of Cambridge, UK.}}

\maketitle

\begin{abstract}
We present several families of total boolean functions which have exact quantum query complexity which is a constant multiple (between $1/2$ and $2/3$) of their classical query complexity, and show that optimal quantum algorithms for these functions cannot be obtained by simply computing parities of pairs of bits. We also characterise the model of nonadaptive exact quantum query complexity in terms of coding theory and completely characterise the query complexity of symmetric boolean functions in this context. These results were originally inspired by numerically solving the semidefinite programs characterising quantum query complexity for small problem sizes. We include numerical results giving the optimal success probabilities achievable by quantum algorithms computing all boolean functions on up to 4 bits, and all symmetric boolean functions on up to 6 bits.
\end{abstract}


\section{Introduction}

Many important quantum algorithms operate in the query complexity model. In this model, the quantity of interest is the number of oracle queries to the input $x \in \{0,1\}^n$ required to compute some (possibly partial) function $f(x)$. Many functions $f$ are now known to admit quantum speed-up in the case where the algorithm is allowed a constant probability of error, and the model of bounded-error quantum query complexity is now relatively well understood. Intriguingly, the model of exact quantum query complexity, where the quantum algorithm must succeed with certainty, seems to be more mysterious.

If $f$ is allowed to be a partial function, it is known that there can be an exponential separation between exact quantum and exact classical query complexity~\cite{deutsch92}, and even between exact quantum and bounded-error classical query complexity~\cite{brassard97}. There are also examples of total functions where exact quantum algorithms allow a speed-up over classical algorithms. In particular, it is known that the parity of $n$ bits can be computed exactly using only $\lceil n/2 \rceil$ quantum queries~\cite{cleve98a,farhi98a}, while any classical algorithm (exact or randomised) must make exactly $n$ queries. However, prior to this work this was the best quantum speed-up known for the exact computation of total boolean functions. Even worse, to the authors' knowledge this was essentially the {\em only} exact quantum speed-up known. Some years ago, Hayes, Kutin and van Melkebeek~\cite{hayes02} gave a quantum algorithm that computes the majority function on $n$ bits exactly using only $n + 1 - w(n)$ queries, where $w(n)$ is the number of 1s in the binary expansion of $n$, but the only quantum ingredient in this algorithm is computing the parity of 2 bits using 1 query. The same applies to quantum algorithms presented by Dubrovska and Mischenko-Slatenkova~\cite{dubrovska06}, and also algorithms by Vasilieva~\cite{vasilieva07,vasilieva09}. In the other direction, it is known that the separation between quantum and classical exact query complexity can be at most cubic~\cite{midrijanis04}, whereas the best known relationship between bounded-error quantum and classical query complexity is a 6th power~\cite{beals01}.

The question of whether the factor of 2 speed-up obtained for computing parity can be beaten is of particular interest because any constant factor speed-up greater than a factor of 2 would in fact give an example of an {\em asymptotic} quantum-classical separation for total functions in the exact model\footnote{We would like to thank Scott Aaronson for stressing this point to us.}, via a construction similar to the Recursive Fourier Sampling problem of Bernstein and Vazirani~\cite{bernstein97}. Indeed, given a boolean function $f$ on $n$ bits, let $f^{(2)}$ be the function defined by $f^{(2)}(x) = f(f(x_1), \dots, f(x_n))$, where $x_1,\dots,x_n$ are $n$-bit strings; and more generally, for $k > 2$, let $f^{(k)}$ be $f$ applied to $n$ independent instances of $f^{(k-1)}$. Then any quantum algorithm which computes $f$ exactly using $q$ queries can be applied recursively to obtain a quantum algorithm computing $f^{(k)}$ exactly using $O((2q)^k)$ queries (the factor of 2 arises from the need to uncompute ``garbage'' left over from each use of the algorithm~\cite{aaronson03}). However, if every exact classical algorithm which computes $f$ exactly must make at least $d$ queries, one can show that any exact classical algorithm for $f^{(k)}$ must make at least $d^k$ queries. If $q < d/2$, this would imply an asymptotic separation.

{\bf Note added.} Following the completion of this work, in a breakthrough result, Ambainis has demonstrated such an asymptotic separation between quantum and classical exact query complexity for a total function~\cite{ambainis12b}. Also, very recently Ambainis, Iraids and Smotrovs~\cite{ambainis13} have developed optimal exact quantum algorithms to determine whether an $n$-bit string has Hamming weight exactly $k$, and to determine whether an $n$-bit string has Hamming weight at least $k$, verifying a conjecture in Section \ref{sec:open}.


\subsection{Our results}

We show that the model of exact quantum query complexity is richer than it may have hitherto appeared. Our results can be divided into analytical and numerical parts. On the analytical side, the main results are as follows.

\begin{itemize}
\item We present several new families of boolean functions $f$ for which the exact quantum query complexity of $f$ is a constant multiple (between $1/2$ and $2/3$) of its classical query complexity, and we show that optimal quantum algorithms for these functions cannot be obtained by simply computing parities of pairs of bits. These separations are based on concatenating small separations obtained for functions on small numbers of bits; indeed, we give optimal exact quantum query algorithms for every boolean function on 3 bits.

\item We give a simple and explicit quantum algorithm for determining whether a 4-bit string has Hamming weight 2, using only 2 queries. Again, this cannot be done merely by computing parities of pairs of input bits. More generally, we give an exact algorithm which distinguishes between inputs of Hamming weight $n/2$ and Hamming weight in the set $\{0,1,n-1,n\}$, for all even $n$, using 2 queries. This generalises the well-known Deutsch-Jozsa problem~\cite{deutsch92} of distinguishing Hamming weight $n/2$ from Hamming weight in $\{0,n\}$.

\item We characterise the model of {\em nonadaptive} exact quantum query complexity in terms of a coding-theoretic quantity. In this setting, all queries to the input must be made up front, in parallel. In contrast to the classical case, there exist non-trivial quantum algorithms in this model. Using our characterisation result, we completely determine the nonadaptive exact quantum query complexity of symmetric boolean functions.
\end{itemize}

These analytical results were inspired by numerical investigations in which we evaluated the best success probability achievable by quantum algorithms computing all boolean functions on up to 4 bits, and all symmetric boolean functions on up to 6 bits. This was achieved using the semidefinite programming approach of Barnum, Saks and Szegedy~\cite{barnum03} to formulate semidefinite programs (SDPs) giving the precise optimal success probability for quantum algorithms using up to $k$ queries, for all $k < n$. We then solved these SDPs numerically using the CVX package~\cite{cvx} for Matlab; the results are given in Section \ref{sec:smallfns} and Appendix \ref{sec:numerical}. Given an SDP solution, one can then write down an {\em explicit} quantum query algorithm with the same parameters; we discuss how this can be done in Section \ref{sec:sdp}. Our analytical results were based on inspecting these algorithms.

Some further highlights from the numerical results are as follows.

\begin{itemize}
\item We conjecture that the exact quantum query complexity of the problem of determining whether an $n$-bit string has Hamming weight exactly $k$ is precisely $\max\{k,n-k\}$. We also conjecture that determining whether an $n$-bit string has Hamming weight at least $k$ has exact quantum query complexity $\max\{k,n-k+1\}$ for $k\ge 1$. Both conjectures hold numerically for all $n\le 6$.

\item In particular, we present strong numerical evidence that the algorithm of Hayes, Kutin and van Melkebeek~\cite{hayes02} for the majority function is not always optimal, by showing that there exists an exact quantum algorithm for computing the majority function on 5 bits using only 3 queries, while their algorithm would require 4 queries.

\item We show numerically that all boolean functions on up to 4 bits, with the exception of functions equivalent to the AND function, have an exact quantum query algorithm using at most 3 queries.

\item We show numerically that no boolean function on up to 5 bits has exact quantum query complexity strictly less than half its exact classical query complexity.
\end{itemize}

Note that H\o yer, Lee and \v{S}palek have previously solved related SDPs (known as the adversary and generalised adversary bounds) for all boolean functions on up to 4 bits~\cite{hoyer07}. These results are also included in~\cite{reichardt08}. These SDPs give lower bounds on bounded-error quantum query complexity but do not characterise it exactly, although the generalised adversary bound does so up to a constant factor~\cite{reichardt11}.


\subsection{Organisation}

The remainder of this paper is organised as follows. After formalising some definitions in Section \ref{sec:def}, in Section \ref{sec:sep} we move on to techniques for finding separations between exact quantum query algorithms, classical algorithms, and quantum algorithms computing parities of input bits. We then discuss in Section \ref{sec:sdp} how the Barnum-Saks-Szegedy SDP can be solved for small problems to give explicit quantum query algorithms. Section \ref{sec:exact2} gives our algorithm for determining whether a 4-bit input has Hamming weight 2. In Section \ref{sec:smallfns} we give optimal exact quantum query algorithms, found by semidefinite programming, for every boolean function on 3 bits. Our results characterising nonadaptive exact quantum query complexity are in Section \ref{sec:nonadapt}, after which we conclude with a discussion of open problems. Two appendices detail our numerical results (including computation of exact quantum query complexity for all 4-bit boolean functions, and all symmetric functions on up to 6 bits), and also give example CVX source code.



\section{Definitions}
\label{sec:def}


\subsection{Generalities and boolean functions}

For any bit-string $x$, $|x|$ will denote the Hamming weight of $x$, and $\bar{x}$ will denote NOT$(x)$. $e_i$ will denote the bit-string with 1 in the $i$'th position, and 0 elsewhere. We use the convention [X] for a quantity which evaluates to 1 if the statement X is true, and 0 otherwise. We will be interested in boolean functions $f:\{0,1\}^n \rightarrow \{0,1\}$, and particularly interested in the following families of functions, all on $n$ bits. PARITY is the function $\text{PARITY}(x) = [|x| \text{ is odd}]$. MAJ is the majority function where MAJ$(x) = [|x| \ge n/2]$. The $\text{EXACT}_k$ function is defined by $\text{EXACT}_k(x) = [|x| = k]$. NAE (``not-all-equal'') evaluates to 0 if all the input bits are equal, otherwise evaluates to 1. Finally, the threshold function $\text{Th}_k(x)$ evaluates to 1 if and only if $|x| \ge k$. All of these are examples of {\em symmetric} boolean functions; a boolean function $f$ is said to be symmetric if $f(x)$ only depends on $|x|$. A non-symmetric function on 3 bits which we will consider is SEL, where SEL$(x_1,x_2,x_3)$ is defined to be equal to $x_2$ if $x_1=0$, and equal to $x_3$ if $x_1=1$.

We say that two boolean functions $f$ and $g$ are {\em isomorphic} if they are equal up to negations and permutations of the input variables, and negation of the output variable. This relationship is sometimes known as NPN-equivalence.


\subsection{Query complexity model}

An exact classical query algorithm to compute a boolean function $f:\{0,1\}^n \rightarrow \{0,1\}$ is described by a decision tree (see e.g.\ \cite{buhrman02}). A decision tree $T$ is a rooted binary tree where each internal vertex has exactly two children, each internal vertex is labelled with a variable $x_i$, $1 \le i \le n$, and each leaf is labelled with 0 or 1. $T$ computes a boolean function as follows: starting with the root, the variable labelling each vertex is queried, and dependent on whether the answer is 0 or 1 the left or right subtree is evaluated. When a leaf is reached, the output is the label of that leaf. The depth of $T$ is the maximal length of a path from the root to a leaf (i.e.\ the worst-case number of queries used on any input). The minimal depth over all decision trees computing $f$ is the exact classical query complexity (aka decision tree complexity) $D(f)$.

We follow what is essentially the standard quantum query complexity model (see e.g.\ \cite{barnum03,hoyer05}). A quantum query algorithm to compute a boolean function $f:\{0,1\}^n \rightarrow \{0,1\}$ is specified by a sequence of unitary operators $U_0,\dots,U_t$ which do not depend on the (initially unknown) input $x$. These unitaries are interspersed with oracle calls $O_x$ (which do depend on the input $x$). The final state of an algorithm that makes $t$ queries, before the final measurement, is given by $U_t O_x U_{t-1} O_x \dots O_x U_0 \ket{0}$. The overall Hilbert space $\mathcal{H}$ used by the quantum query algorithm is split into three subspaces $\mathcal{H}_{\text{in}} \otimes \mathcal{H}_{\text{work}} \otimes \mathcal{H}_{\text{out}}$. For boolean functions, $\mathcal{H}_{\text{out}}$ is a single qubit, whereas the workspace $\mathcal{H}_{\text{work}}$ is of arbitrary size. The oracle $O_x$ acts only on $\mathcal{H}_{\text{in}}$ by mapping $\ket{i} \mapsto (-1)^{x_i} \ket{i}$ for some hidden bit string $x$. $\mathcal{H}_{\text{in}}$ is $n+1$ dimensional and indexed by basis vectors $\ket{0},\dots,\ket{n}$; a query to $x_0$ always returns 0. The final step of the algorithm is always simply to measure the $\mathcal{H}_{\text{out}}$ register and return the outcome. We say that the algorithm computes $f$ within error $\epsilon$ if, on input $x$, the algorithm returns $f(x)$ with probability at least $1-\epsilon$ for all $x$. The exact quantum query complexity $Q_E(f)$ is the minimum number of queries used by any quantum algorithm which computes $f(x)$ exactly for all $x$.

Note that if boolean functions $f$ and $g$ are isomorphic, $D(f)=D(g)$ and $Q_E(f)=Q_E(g)$.


\section{Separating exact quantum and classical query complexity}
\label{sec:sep}

We observe that a fixed function demonstrating a separation between exact quantum and classical query complexity, even a small additive constant, can give rise to a constant factor multiplicative separation for an asymptotically growing family of functions. For a boolean function $f:\{0,1\}^n \rightarrow \{0,1\}$, let $\deg(f)$ be the degree of the multilinear polynomial representing $f$ over the reals (see e.g.~\cite{nisan94,buhrman02} for a precise definition). Then we have the following easy proposition.

\begin{prop}
\label{prop:growing}
Let $f:\{0,1\}^k \rightarrow \{0,1\}$ be a boolean function on $k$ bits such that $Q_E(f) = q$ and $\deg(f) = d$. Also let $g:\{0,1\}^n \rightarrow \{0,1\}$ be a boolean function such that $\deg(g)=n$. Define the family of functions $f_n:\{0,1\}^{nk} \rightarrow \{0,1\}$ as follows: divide the $nk$ input bits into blocks $b_1,\dots,b_n$ of $k$ bits each, and set $f_n(x_1,\dots,x_{nk}) = g(f(b_1),f(b_2),\dots,f(b_n))$. Then $d n / 2 \le Q_E(f_n) \le q n$ and $D(f_n) \ge dn$.
\end{prop}

\begin{proof}
An exact quantum query algorithm for $f_n$ using $q n$ queries can be obtained simply by computing $f(b_i)$ for each block $b_i$ individually, and then computing $g(f(b_1),\dots,f(b_n))$ without any further queries, so $Q_E(f_n) \le q n$. On the other hand, $\deg(f_n) = \deg(f) \deg(g) = d n$. As $D(f_n) \ge \deg(f_n)$~\cite{nisan94} and $Q_E(f_n) \ge \deg(f_n)/2$~\cite{beals01}, $D(f_n) \ge d n$ and $Q_E(f_n) \ge dn/2$.
\end{proof}

Natural examples of functions $g$ to which Proposition \ref{prop:growing} can be applied are AND and OR. Thus an example of an exact query complexity separation we obtain from our results on small boolean functions is the following:

\begin{thm}
Let EXACT$_2^\ell$ be the boolean function on $4\ell$ bits defined as follows. Split the input $x$ into consecutive blocks $b_1,\dots,b_\ell$ containing 4 bits each, and set EXACT$_2^\ell(x_1,\dots,x_{4\ell})=1$ if each block $b_i$ contains exactly two 1s. Then $Q_E(\text{EXACT}_2^\ell) = 2\ell$ and $D(\text{EXACT}_2^\ell) = 4\ell$.
\end{thm}

\begin{proof}
It is easy to verify that the polynomial degree of the $\text{EXACT}_2$ function on 4 bits is equal to 4; we prove that $Q_E(\text{EXACT}_2) = 2$ in Section \ref{sec:exact2}. The claim then follows from Proposition \ref{prop:growing}.
\end{proof}


\subsection{Quantum algorithms based on parity queries}

It is well-known that quantum computers can evaluate the parity of two input bits exactly using only one query~\cite{cleve98a}. Thus a non-trivial class of exact quantum query algorithms consists of classical decision trees, each of whose vertices corresponds to a query either to an individual bit of the input, or to the parity of two input bits. We now give a lower bound on the number of queries used by such algorithms (indeed, a more general class of algorithms). Define {\em parity decision trees} to be the modification of decision trees where each internal vertex $v$ is labelled with a subset $S_v$ of the input variables. When $v$ is reached, the parity $\bigoplus_{i \in S_v} x_i$ is computed. If the answer is 0, the left subtree is evaluated; if the answer is 1, the right subtree is evaluated. Standard decision trees are simply the special case of parity decision trees where $|S_v|=1$ for all $v$, while decision trees based on the use of the quantum algorithm for PARITY are parity decision trees such that $|S_v| \le 2$ for all $v$.

We then have the following simple result.

\begin{prop}
\label{prop:f2degree}
Let $f:\{0,1\}^n \rightarrow \{0,1\}$ be a boolean function, and let $d$ be the degree of $f$ as an $n$-variate polynomial over $\F_2$. Then any parity decision tree which computes $f$ must have depth at least $d$.
\end{prop}

\begin{proof}
We will show by induction that any decision tree on parities which has depth $D$ gives rise to a degree $D$ polynomial over $\F_2$. The function computed by any such tree can be written as $p T_0 + (1+p)T_1$ for some degree 1 polynomial $p$ over $\F_2$ and decision trees $T_0$, $T_1$ of depth at most $D-1$. Therefore, the degree of the polynomial obtained by repeating this procedure recursively is at most $D$. If the tree computes $f$ on every input, this polynomial must be equal to $f$, and hence be degree $d$. Thus $D \ge d$.
\end{proof}


\section{Quantum query algorithms from semidefinite programming}
\label{sec:sdp}

In this section we discuss, following~\cite{barnum03}, how quantum query complexity can be evaluated as a semidefinite programming (SDP) problem, given in Definition \ref{dfn:sdp} below. In this definition, $\circ$ is the Hadamard (entrywise) product of matrices.

\boxdfn{Quantum query complexity primal SDP}{
\label{dfn:sdp}
\hspace{-0.5em}For a given boolean function $f:\{0,1\}^n \rightarrow \{0,1\}$ and a given integer $t$, find a sequence of $2^n$-dimensional real symmetric matrices $(M_i^{(j)})$, where $0 \le i \le n$ and $0 \le j \le t-1$, and $2^n$-dimensional real symmetric matrices $\Gamma_0$, $\Gamma_1$, satisfying the following constraints:
\bea
\label{eq:startconstr} \sum_{i=0}^n M_i^{(0)} &=& E_0\\
\label{eq:runconstr} \sum_{i=0}^n M_i^{(j)} &=& \sum_{i=0}^n E_i \circ M_i^{(j-1)}\;\text{(for $1 \le j \le t-1$)}\\
\label{eq:lastconstr} \Gamma_0 + \Gamma_1 &=& \sum_{i=0}^n E_i \circ M_i^{(t-1)}\\
\label{eq:outconstr0} F_0 \circ \Gamma_0 &\ge& (1-\epsilon)F_0\\
\label{eq:outconstr1} F_1 \circ \Gamma_1 &\ge& (1-\epsilon)F_1,
\eea
where $E_i$ is the matrix $\bracket{x}{E_i}{y} = (-1)^{x_i+y_i}$, and $F_0$ and $F_1$ are diagonal matrices defined by $\bracket{x}{F_z}{x} = 1$ if and only if $f(x) = z$, and $\bracket{x}{F_z}{x} = 0$ otherwise.
}

The following important characterisation will be the key to many of our results.

\begin{thm}[Barnum, Saks and Szegedy~\cite{barnum03}]
There is a quantum query algorithm that uses $t$ queries to compute a function $f:\{0,1\}^n \rightarrow \{0,1\}$ within error $\epsilon$ if and only if the SDP of Definition \ref{dfn:sdp} is feasible.
\end{thm}

Therefore, if one minimises $\epsilon$ subject to these semidefinite constraints, one obtains the lowest possible error that can be achieved by a quantum algorithm which computes $f$ using $t$ queries.


\subsection{A prescription for quantum algorithms}

It is perhaps not immediately obvious how, given a solution to the semidefinite program of Definition \ref{dfn:sdp}, to produce a quantum query algorithm with the same parameters. This was implicit in \cite{barnum03}; we now spell out explicitly how it can be done. We will use the following standard lemma from linear algebra (see e.g.\ \cite{horn85}).

\begin{lem}
\label{lem:transition}
Let $S = (\ket{\psi_i})$ and $T = (\ket{\phi_j} )$ be two sequences of $m$ vectors of the same dimension. Define $\Psi = \sum_i \ket{\psi_i}\bra{i}$, $\Phi = \sum_j \ket{\phi_j}\bra{j}$. Then there is a unitary $U$ such that $U \ket{\phi_i} = \ket{\psi_i}$ for all $i$ if and only if $\Psi^\dag \Psi = \Phi^\dag \Phi$. If such a $U$ exists, it can be written down as follows. Let $V$ and $W$ be any isometries satisfying $\Psi = V \sqrt{\Psi^\dag \Psi}$, $\Phi = W \sqrt{\Phi^\dag \Phi}$ (i.e.\ isometries occurring in polar decompositions of $\Psi$, $\Phi$), and complete $V$ and $W$ to unitary matrices $V'$ and $W'$. Then $U=V' (W')^\dag$.

\end{lem}

Given a set of matrices $M_i^{(j)}$ as in Definition \ref{dfn:sdp}, one can derive an explicit quantum query algorithm completely mechanically, as follows. Use a workspace $\mathcal{H}_{\text{work}}$ of $n$ qubits, and ignore the output qubit for the time being. Let the initial state be $\ket{0}\ket{0}$, and let the state of the system at time $j$ (i.e.\ after $j$ queries have been made, and just before the $(j+1)$'st query is made) be $\ket{\psi_x^{(j)}} = \sum_{i=0}^n \ket{i}\ket{\psi_{x,i}^{(j)}}$, where $\ket{\psi_{x,i}^{(j)}}$ is a subnormalised state in the Hilbert space $\mathcal{H}_{\text{work}}$.

We will {\em define} the state at time $0 \le j \le t-1$ in terms of the matrices $M_i^{(j)}$ occurring in a solution to the SDP by setting
\be \label{eq:psixi} \ket{\psi_{x,i}^{(j)}} = \sqrt{M_i^{(j)}}\ket{x}, \ee
where $\sqrt{M_i^{(j)}}$ is the unique positive semidefinite square root of $M_i^{(j)}$. It is perhaps not immediately clear that following this prescription leads to $\ket{\psi_x^{(j)}}$ being normalised, let alone the sequence $\ket{\psi_x^{(0)}},\ket{\psi_x^{(1)}},\dots,\ket{\psi_x^{(t-1)}}$ corresponding to a valid quantum query algorithm for all $x$; however, we will now see that this is indeed the case.

For any $1 \le j \le t$, define
\[ \ket{\phi_x^{(j)}} = \sum_{i=0}^n (-1)^{x_i} \ket{i}\ket{\psi_{x,i}^{(j-1)}}, \]
and set $\ket{\phi_x^{(0)}} = \ket{0}\ket{0}$. Also define the matrices obtained by concatenating the vectors as columns,
\[ \Psi^{(j)} = \sum_{x \in \{0,1\}^n} \ket{\psi_x^{(j)}}\bra{x},\,\,\Phi^{(j)} = \sum_{x \in \{0,1\}^n} \ket{\phi_x^{(j)}} \bra{x}. \]
The vectors $\ket{\phi_x^{(j)}}$ represent the state of the system immediately {\em after} the $j$'th oracle call, and $\ket{\phi_x^{(0)}}$ is the initial state of the system. We would like to find unitaries $U_0,\dots,U_{t-1}$ mapping $\ket{\phi_x^{(j)}} \mapsto \ket{\psi_x^{(j)}}$ for all $x$. By Lemma \ref{lem:transition}, such a sequence of unitaries will exist if
\[ (\Psi^{(j)})^\dag \Psi^{(j)} = (\Phi^{(j)})^\dag \Phi^{(j)}. \]
But observe that for any $0 \le j \le t-1$,
\[ \bracket{x}{(\Psi^{(j)})^\dag \Psi^{(j)}}{y} = \ip{\psi_x^{(j)}}{\psi_y^{(j)}} = \sum_{i=0}^n \ip{\psi_{x,i}^{(j)}}{\psi_{y,i}^{(j)}} = \sum_{i=0}^n \bracket{x}{M_i^{(j)}}{y}, \]
so $(\Psi^{(j)})^\dag \Psi^{(j)} = \sum_{i=0}^n M_i^{(j)}$. Similarly, for any $1 \le j \le t$,
\beas
\bracket{x}{(\Phi^{(j)})^\dag \Phi^{(j)}}{y} &=& \ip{\phi_x^{(j)}}{\phi_y^{(j)}} = \sum_{i=0}^n (-1)^{x_i+y_i} \ip{\psi_{x,i}^{(j-1)}}{\psi_{y,i}^{(j-1)}}\\
&=&  \sum_{i=0}^n (-1)^{x_i+y_i} \bracket{x}{M_i^{(j-1)}}{y} = \bracket{x}{\left(\sum_{i=0}^n E_i \circ M_i^{(j-1)}\right)}{y}.
\eeas
As, by constraint (\ref{eq:runconstr}) in the SDP, $\sum_{i=0}^n M_i^{(j)} = \sum_{i=0}^n E_i \circ M_i^{(j-1)}$ for all $1 \le j \le t-1$, this implies by Lemma \ref{lem:transition} that for each $j \ge 1$ there exists a unitary $U_j$ such that $U_j\ket{\phi_x^{(j)}} = \ket{\psi_x^{(j)}}$ for all $x$, and this $U_j$ can be determined explicitly from Lemma \ref{lem:transition}. In the case $j=0$, $(\Phi^{(j)})^\dag \Phi^{(j)} = E_0$, and hence SDP constraint (\ref{eq:startconstr}) implies the existence of a $U_0$ such that $U_0\ket{\phi_x^{(0)}} = \ket{\psi_x^{(0)}}$ for all $x$.

The final constraint we need to satisfy is that the algorithm outputs the correct result. Define the final state of the system on input $x$ (just before the output qubit is measured) to be
\[ \ket{\gamma_x} = \ket{0}_{\text{in}}(\sqrt{\Gamma_0}\ket{x})_{\text{work}}\ket{0}_{\text{out}}  + \ket{0}_{\text{in}} (\sqrt{\Gamma_1}\ket{x})_{\text{work}}\ket{1}_{\text{out}}. \]
Then
\[ \ip{\gamma_x}{\gamma_y} = \bracket{x}{\Gamma_0}{y} + \bracket{x}{\Gamma_1}{y} = \bracket{x}{\left(\sum_{i=0}^n E_i \circ M_i^{(t-1)}\right)}{y} \]
by SDP constraint (\ref{eq:lastconstr}), so by a similar argument there exists a $U_t$ such that $U_t\ket{\phi_x^{(t)}} = \ket{\gamma_x}$ for all $x$. Measuring the output qubit gives the answer 0 with probability $\bracket{x}{\Gamma_0}{x}$, which by constraint (\ref{eq:outconstr0}) is at least $1-\epsilon$ when $f(x)=0$. Similarly, by constraint (\ref{eq:outconstr1}) we obtain the answer 1 with probability at least $1-\epsilon$ when $f(x)=1$.

Observe that we have some freedom in our choice of states $\ket{\psi_{x,i}^{(j)}}$; while eqn.\ (\ref{eq:psixi}) gives one choice which always works, it would suffice to pick any states such that $\ip{\psi_{x,i}^{(j)}}{\psi_{y,i}^{(j)}} = \bracket{x}{M_i^{(j)}}{y}$. In particular, if the rank of $M_i^{(j)}$ is upper bounded by $r$ for all $i,j$, one can choose states of dimension $r$ throughout. This would reduce the size of the $\mathcal{H}_{\text{work}}$ register from $n$ qubits to $\lceil \log_2 r \rceil$ qubits. Also observe that without loss of generality all states and unitaries occurring in a quantum query algorithm can be taken to be real.


\section{EXACT$_2$}
\label{sec:exact2}

We now give a simple and explicit quantum algorithm for evaluating the EXACT$_2$ function on 4 bits using only 2 quantum queries. This algorithm was originally inspired by numerically solving the SDP discussed in the previous section. The algorithm does not use any workspace (or even an output register), and hence operates solely on the 5-dimensional input register indexed by basis states $\{\ket{0},\dots,\ket{4}\}$. Define a unitary matrix $U$ by

\[
U=\frac{1}{2}\left( \begin{matrix} 
0 & 1 & 1 & 1 & 1\\ 
1 & 0 & 1 & \omega & \omega^2\\
1 & 1 & 0 & \omega^2 & \omega\\
1 & \omega & \omega^2 & 0 & 1\\
1 & \omega^2 & \omega & 1 & 0
\end{matrix} \right),
\]
where $\omega=e^{2 \pi i/3}$ is a complex cube root of 1. We begin in the state
\[ \ket{\psi}=\frac{1}{2}\sum_{i=1}^4 \ket{i} \]
and then apply $O_x$, then $U$, then $O_x$ again. Finally, we perform the measurement consisting of a projection onto the state $\ket{\psi}$ and its orthogonal complement. If the outcome is $\ket{\psi}$, we output 1, and otherwise 0.

The claim is that $V_x:=O_x U O_x$ leaves $\ket{\psi}$ unchanged, up to a phase factor, when $x$ has Hamming weight 2, and otherwise maps $\ket{\psi}$ into a subspace orthogonal to $\ket{\psi}$. To see that the claim is correct, note first that $U\ket{\psi}=\ket{0}$, since $1+\omega+\omega^2=0$. But for $x=0000$, $O_x$ is the identity. Thus
\[
V_{0000}\ket{\psi}=\ket{0}.
\] 
Similarly, for $x=1111$ we have $V_{1111}\ket{\psi}=-\ket{0}$. For $x=1000$, $O_x\ket{\psi}=\ket{\psi}-\ket{1}$. So $V_x\ket{\psi}=\ket{0}-O_x U\ket{1}$. But the coefficient of $\ket{1}$ in $U\ket{1}$ is zero, so $O_x$ leaves $U\ket{1}$ unchanged and we have
\[
V_{1000}\ket{\psi}=\ket{0}-U\ket{1}.
\] 
Similar results hold for the other weight 1 strings $x$.

For $x=1100$, $O_x\ket{\psi}=\ket{\psi}-\ket{1}-\ket{2}$, and 
\begin{align*}
U(\ket{\psi}-\ket{1}-\ket{2})&=\ket{0}-\frac{1}{2} \left( 2\ket{0}+\ket{1}+\ket{2}+(\omega+\omega^2)(\ket{3}+\ket{4}) \right)\\
&=\frac{1}{2}(-\ket{1}-\ket{2}+\ket{3}+\ket{4}).
\end{align*}
Applying $O_x$ once more we get
\[
V_{1100}\ket{\psi}=\ket{\psi}.
\]
We get the same result for other strings of weight 2, possibly with a
phase factor. For instance
\[
V_{1001}\ket{\psi}=\omega^2 \ket{\psi}.
\]
Given a string $x$ of weight 3, we can flip all the
bits and the oracle behaves identically, up to a change of sign on the space spanned by $\ket{1},\dots,\ket{4}$. For example,
\[
V_{0111}\ket{\psi}=-\ket{0}-U\ket{1},
\] 
and the other strings of weight 3 are similar.

Thus, for $x$ such that $|x| \neq 2$, $V_x\ket{\psi}$ lies in the span of
$\ket{0}$ and $U\ket{i}$ for $i=1,2,3,4$. However $\langle
\psi|0\rangle=0$ and $\langle
\psi|U|i\rangle=\frac{1}{4}(1+\omega+\omega^2)=0$ for $i=1,2,3,4$. So
this subspace is orthogonal to $\ket{\psi}$, proving our claim.

\subsection{Distinguishing weights 0 and 1 from balanced strings}
\label{sec:balanced}

We would ideally like to understand the number of queries required to solve the EXACT$_k$ function on $n$ input bits, for all $n$ and $k$. As an intermediate goal generalising our solution for $n=4$, one can ask for a two-query algorithm, for any even $n$, that distinguishes strings of weight 0 and 1 from strings of weight $n/2$. We take the input space $\mathcal{H}_{\text{in}}$, with basis vectors $\ket{i}$, $i=0,1, \ldots, n$, and tensor it with an ancilla space.
In the ancilla space we select vectors $\ket{a_i}$, $i=1,\ldots, n$, of
unit length with inner products $\ip{a_i}{a_j}=c$, for $i \ne j$, for some real $c$ which will be chosen later. We also select some orthogonal vector $\ket{0}$ in the ancilla space.

The oracle acts on the input space, so
$O_x\ket{i}\ket{j}=(-1)^{x_i}\ket{i}\ket{j}$ for any $i,j$.  Recall that
$e_i$ is the string with a 1 at position $i$ and 0's elsewhere, and let
$b \in \{0,1\}^n$ be an arbitrary ``balanced'' string with $|b|=n/2$. The
algorithm starts with the state $\ket{\phi}=\sum_{i=1}^n
\ket{i}\ket{0}$ (we keep $\ket{\phi}$ unnormalised throughout for simplicity).  As before, $V_x:=O_x U O_x$, with $U$ to be
defined shortly. The aim is to show that, for any balanced $b$, $V_b\ket{\phi}$ is orthogonal
to the subspace spanned by $V_{0^n}\ket{\phi},V_{e_1}\ket{\phi},\dots,V_{e_n}\ket{\phi}$, which
thus discriminates between balanced strings and strings of weight at most 1.

Now define $U$ by its action on the states $\ket{\tau_i}=O_{e_i}\ket{\phi}$ by
\begin{align*}
U\ket{\tau_i}=\alpha \ket{00}+\beta\sum_{j=1}^n \ket{j}\ket{a_{j-i+1}}+\gamma\ket{i}\ket{a_1},
\end{align*}
where $\alpha$, $\beta$ and $\gamma$ are real and $j-i$ is taken mod $n$, so $j-i+1 \in \{1,\dots,n\}$. This will be an isometry (which can be extended to a unitary on the whole tensor product space) if 
\begin{align*}
\alpha^2+(n-1)\beta^2+(\gamma+\beta)^2&=n,\\
\alpha^2+(n-2)\beta^2c+2(\gamma+\beta)\beta c&=n-4.
\end{align*}
We find
\begin{align}
\label{e1} V_{e_i}\ket{\phi}=O_{e_i}U\ket{\tau_i}=\alpha \ket{00}+\beta\sum_{j=1}^n \ket{j}\ket{a_{j-i+1}}-(\gamma+2\beta)\ket{i}\ket{a_1},
\end{align}
and
\begin{align}
\nonumber (n-2) U\ket{\phi}&=U \left(\sum_{i=1}^n \ket{\tau_i}\right)\\
\label{Uphi} &=n\alpha\ket{00}+\beta\left(\sum_{i=1}^n \ket{i} \right)\left( \sum_{i=1}^n \ket{a_i} \right)+ \gamma \left( \sum_{i=1}^n \ket{i} \right)  \ket{a_1},
\end{align}
and
\begin{align}
\nonumber 2V_b\ket{\phi}&=O_bU \sum_{i=1}^n (-1)^{b_i+1}\ket{\tau_i}\\
\label{Ub} &=\beta \sum_{j=1}^n \ket{j} \sum_{i=1}^n(-1)^{b_i+b_j+1}\ket{a_{j-i+1}}-\gamma \left( \sum_{i=1}^n\ket{i} \right) \ket{a_1}.
\end{align}
Note that the coefficient of $\ket{j}\ket{a_1}$ in the first term of the preceding equation is $\beta(-1)^{2b_i+1}=-\beta$. Using this fact, (\ref{Uphi}) and (\ref{Ub}) imply
\begin{align*}
(n-2)\bra{\phi}V_bU\ket{\phi}=\beta^2 \sum_{i=1}^n \sum_{j=1}^n (-1)^{b_i+b_j+1}-n\gamma\beta(2-(n-2)c)-n\gamma^2,
\end{align*}
and since the first term vanishes, the condition for $\bra{\phi}V_bU\ket{\phi}=0$ is
\begin{align}
\label{weight0} \gamma \beta (2+(n-2)c)+\gamma^2=0.
\end{align}
Similarly, from (\ref{e1}) and (\ref{Ub}) we deduce that $\bra{\phi}V_bV_{e_i}\ket{\phi}=0$ if and only if
\begin{align}
\label{weight1} 2\beta^2(1-c)+\gamma \beta (2-nc)+\gamma^2=0.
\end{align}
From (\ref{weight1}) and (\ref{weight0}) we find
\begin{align*}
(n-1)(n-2)c^2+(2n-3)c+1=0,
\end{align*}
which has roots $c=-1/(n-1), -1/(n-2)$. It is not possible to find a
set of unit length vectors $\ket{a_i}$ with
$\ip{a_i}{a_j}=-1/(n-2)$ for $i \ne j$; this follows from the fact
that the associated Gram matrix is not positive
semidefinite. However, this condition is not violated if
$\ip{a_i}{a_j}=-1/(n-1)$, and in fact we can choose the $a_i$ to
be the normalised vectors from the centre of a regular $(n-1)$-dimensional
simplex to its vertices, i.e. the vectors consisting of
\[
\sqrt{\frac{n}{n-1}}\left(-\frac{1}{n}, -\frac{1}{n}, \ldots , -\frac{1}{n},\frac{n-1}{n}\right)
\]
and its permutations. 

With $c=-1/(n-1)$ we find
\begin{align*}
\alpha^2&=\frac{n^3-6n^2+12n-12}{n(n-2)},\\
\beta&=\frac{2(n-1)}{n\sqrt{n-2}},\\
\gamma&=\frac{-2}{\sqrt{n-2}}.
\end{align*}

We have therefore shown that, with this choice of $U$, the algorithm correctly distinguishes between balanced strings and strings of Hamming weight at most 1. One can verify that the same argument goes through for strings of Hamming weight at least $n-1$. Thus the algorithm correctly distinguishes between inputs of  weight $n/2$ and weight in the set $\{0,1,n-1,n\}$.


\section{Exact quantum query algorithms for small functions}
\label{sec:smallfns}


Having whetted our appetite with the EXACT$_2$ problem, we now turn to quantum algorithms for other small boolean functions. For each function on $n$ bits we considered, we calculated, via the SDP of Definition~\ref{dfn:sdp}, the best possible success probability achievable by quantum algorithms making $t$ queries, for $t=1,\dots,n-1$ (any function can clearly be computed exactly using $n$ queries). We did this for all functions on up to 4 input bits, and for all symmetric functions on up to 6 bits. We used the convex optimisation package CVX~\cite{cvx} for Matlab, which allows optimisation problems to be specified in a simple and intuitive way; see Appendix~\ref{sec:cvxcode} for source code. The CVX package allows the choice of underlying solvers SeDuMi and SDPT3. We used SeDuMi for the results given below, and also checked the results with SDPT3, which gave the same values up to a difference of at most $0.001$. The numerical results for functions on 4 bits, and symmetric functions on up to 6 bits, are deferred to Appendix~\ref{sec:numerical}.

Note that there is a basic issue with calculating {\em exact} quantum query complexity numerically, which is that one receives a numerical solution from the SDP solver, which is not exact. If the SDP solver claims that there exists a quantum query algorithm that computes some function $f$ using $k$ queries with success probability at least $0.999$, for example, one cannot be sure that this algorithm is actually exact. In the case of all functions on up to 3 bits, we therefore give explicit optimal {\em exact} quantum query algorithms. These algorithms were obtained by a somewhat laborious process of taking the numerically obtained (real-valued, approximate) solutions to the SDP and using these as a guide to find exact solutions.

For completeness, we begin by giving optimal exact quantum query algorithms for all boolean functions of 1 and 2 bits. In what follows, the tables are indexed by function ID; the ID of each function is the integer obtained by converting its truth table from binary. Columns give the optimal success probability that can be achieved by quantum algorithms making $1,\dots,n-1$ queries. Entries are starred when there is a {\em nonadaptive} exact quantum algorithm using that number of queries (see Section \ref{sec:nonadapt}).


\subsection{Functions of up to 2 bits}

Up to isomorphism, the only non-constant function of 1 bit is $f(x_1) = x_1$, which clearly requires exactly one query. In the case of 2 bits, there are two classes of functions.

\begin{center}
\begin{tabular}{|c|c|c|c|c|}
\hline {\bf ID} & {\bf Function} & {\bf 1 query}\\
\hline 1 & $x_1 \wedge x_2$ & 0.900\\
\hline 6 & $x_1 \oplus x_2$ & 1*\\
\hline
\end{tabular}
\end{center}

An optimal quantum algorithm for the function $x_1 \oplus x_2$ proceeds as follows~\cite{cleve98a}. Input the state $\frac{1}{\sqrt{2}}(\ket{1}+\ket{2})$ into the oracle to produce $\frac{1}{\sqrt{2}}((-1)^{x_1}\ket{1} + (-1)^{x_2}\ket{x_2})$. Perform a Hadamard gate (with respect to the basis $\{\ket{1},\ket{2}\}$), measure in the basis $\{\ket{1},\ket{2}\}$, and output 0 if ``1'' is measured, and 1 if ``2'' is measured. It is easy to see that this algorithm succeeds with certainty.


\subsection{Functions of 3 bits}

The following table lists the optimal success probability that can be achieved by quantum algorithms computing all boolean functions depending on 3 bits, up to isomorphism.
\begin{center}
\begin{tabular}{|c|c|c|c|c|c|c|}
\hline {\bf ID} & {\bf Function} & {\bf 1 query} & {\bf 2 queries}& {\bf $\F_2$ deg.} & {\bf D(f)} \\
\hline 1 & $x_1 \wedge x_2 \wedge x_3$ & 0.800 & 0.980 & 3 & 3\\
6 & $x_1 \wedge (x_2 \oplus x_3)$ & 0.667 & 1* & 2 & 3\\
7 & $x_1 \wedge (x_2 \vee x_3)$ & 0.773 & 1 & 3 & 3\\
22 & EXACT$_2$ & 0.571 & 1 & 3 & 3\\
23 & MAJ & 0.667 & 1 & 2 & 3\\
30 & $x_1 \oplus (x_2 \vee x_3)$ & 0.667 & 1 & 2 & 3\\
53 & SEL$(x_1,x_2,x_3)$ & 0.854 & 1 & 2 & 2\\
67 & $(x_1 \wedge x_2) \vee (\bar{x_1} \wedge \bar{x_2} \wedge x_3)$ & 0.773 & 1 & 3 & 3\\
105 & PARITY & 0.500 & 1* & 1 & 3\\
126 & NAE & 0.900 & 1* & 2 & 3\\
\hline
\end{tabular}
\end{center}

Observe that the AND function requires 3 queries to be computed exactly; in fact, it has been known for some time that AND on $n$ bits has $Q_E(\text{AND})=n$~\cite{beals01}. For most of the other functions, an optimal exact quantum algorithm is easy to determine, based only on classical queries and computing the parity of two bits using one query:

\begin{itemize}
\item $x_1 \wedge (x_2 \oplus x_3)$: Query $x_1$ and evaluate $x_2 \oplus x_3$ using one query.
\item MAJ: First evaluate $x_1 \oplus x_2$. If the answer is 1, then output $x_3$, otherwise output $x_1$. This works because if $x_1$ and $x_2$ are different, then there will be at least two 1's in total if and only if $x_3$ is 1. If $x_1$ and $x_2$ are the same, there will be at least two 1's if and only if $x_1$ is 1. 
\item $x_1 \oplus (x_2 \vee x_3)$: This function is equivalent to $(\bar{x_3} \wedge x_2) \vee (x_3 \wedge (x_1 \oplus x_2))$. So query $x_3$ first, then either query $x_2$ or $x_1 \oplus x_2$.
\item SEL$(x_1,x_2,x_3)$: Query $x_1$ first, then either output $x_2$ or $x_3$.
\item PARITY: Evaluate $x_1 \oplus x_2$, query $x_3$, take the exclusive OR of the two.
\item NAE: This function is equivalent to $(x_1 \oplus x_2) \vee (x_1 \oplus x_3)$.
\end{itemize}

However, the three remaining functions (EXACT$_2$, $x_1 \wedge (x_2 \vee x_3)$ and $(x_1 \wedge x_2) \vee (\bar{x_1} \wedge \bar{x_2} \wedge x_3)$) do not have such straightforward optimal algorithms. Indeed, by Proposition \ref{prop:f2degree}, they cannot be computed using 2 queries by any algorithm which is a decision tree on parity queries. In the case of EXACT$_2$, we obtain an optimal algorithm by appending an additional zero bit and computing EXACT$_2$ on 4 bits (see Section \ref{sec:exact2}). We now give quantum query algorithms for the two remaining functions. Rather than writing out the unitary operators arising in the algorithm explicitly, we simply give expressions for matrices forming a exact solution to the query complexity SDP. We stress that, given these matrices, one can follow the procedure of Section \ref{sec:sdp} to find an explicit quantum algorithm completely mechanically. The matrices are fully specified by their non-zero eigenvalues and eigenvectors.


\subsubsection{$x_1 \wedge (x_2 \vee x_3)$}

\[
\begin{array}{|c|c|c|}
\hline
\text{\bf Matrix} & \text{\bf Eigenvalues} & \text{\bf Eigenvectors} \\
\hline M_0^{(0)}, M_1^{(0)}, & 2 & (1,1,1,1,1,1,1,1) \\
 M_2^{(0)}, M_3^{(0)} & & \\
\hline M_0^{(1)} & 3/2 & (1,1,1,1,0,0,0,0) \\
& 1 & \{(-1,1,-1,1,0,2,0,2),(0,-1,1,0,0,-1,1,0)\} \\
\hline M_1^{(1)} & 1 & (1,0,0,-1,2,1,1,0) \\
& 1/2 & (-1,0,-1,0,0,1,0,1) \\
\hline M_2^{(1)} & 1 & (1,0,0,-1,2,1,1,0) \\
& 1/2 & (-1,-1,0,0,0,0,1,1) \\
\hline M_3^{(1)} & 3/4 & \{(0,1,0,1,0,1,0,1),(0,-1,1,0,0,-1,1,0)\} \\
\hline \Gamma_0 & 5/2 & (3,2,2,3,2,0,0,0)\\
& 1 & \{(0,-1,0,0,1,0,0,0),(0,-1,1,0,0,0,0,0)\}\\
& 1/2 & (-1,0,0,1,0,0,0,0)\\
\hline \Gamma_1 & 3/2 & \{(0,0,0,0,0,1,0,1),(0,0,0,0,0,-1,1,0)\}\\
 \hline
\end{array}
\]


\subsubsection{$(x_1 \wedge x_2) \vee (\bar{x_1} \wedge \bar{x_2} \wedge x_3)$}
For conciseness, set $\alpha^{\pm} = -1+\frac{1}{2}(5\pm\sqrt{5})$ in the following table.
\[
\begin{array}{|c|c|c|}
\hline
\text{\bf Matrix} & \text{\bf Eigenvalues} & \text{\bf Eigenvectors} \\
\hline M_0^{(0)}, M_1^{(0)}, & 2 & (1,1,1,1,1,1,1,1)\\
M_2^{(0)}, M_3^{(0)} & & \\
\hline M_0^{(1)} & 1 & (-2,-1,-1,0,-1,0,0,1) \\
& 3/4 & (0,0,-1,-1,2,2,1,1) \\
& 1/4 & (0,0,1,1,0,0,1,1)\\
\hline M_1^{(1)} & 1 & (-2,-1,-1,0,-1,0,0,1) \\
& 3/4 & (0,-1,2,1,-1,-2,1,0) \\
& 1/4 & (0,-1,0,-1,1,0,1,0)\\
\hline M_2^{(1)} & 1 & (0,1,1,2,1,2,2,3) \\
& 3/4 & (0,-1,0,-1,1,0,1,0) \\
& 1/4 & (0,-1,2,1,-1,-2,1,0)\\
\hline M_3^{(1)} & 1 & (0,3,-1,2,-1,2,-2,1) \\
& 3/4 & (0,0,1,1,0,0,1,1) \\
& 1/4 & (0,0,-1,-1,2,2,1,1) \\
\hline \Gamma_0 & \frac{1}{4}(5+\sqrt{5}) & (1+\sqrt{5},0,\alpha^+,1,\alpha^+,1,0,0)\\
& 3/2 & (0,0,0,-1,0,1,0,0)\\
& 1 & (0,0,-1,0,1,0,0,0)\\
& \frac{1}{4}(5-\sqrt{5}) & (1-\sqrt{5},0,\alpha^-,1,\alpha^-,1,0,0)\\
\hline \Gamma_1 & 3/2 & \{(0,-1,0,0,0,0,0,1),(0,1,0,0,0,0,1,0)\}\\
 \hline
\end{array}
\]


\section{Nonadaptive exact quantum query complexity}
\label{sec:nonadapt}

We now turn to a very restricted model of exact query complexity, in which the algorithm's queries are required to be nonadaptive. In other words, the choice of which input variables to query cannot depend on the result of previous queries, so the algorithm must choose which variables to query at the start\footnote{In some sense, all quantum query algorithms are nonadaptive, as the choice of unitaries applied in the algorithm does not depend on the input. However, the weight placed on queries to different input bits throughout the algorithm does in general depend on the input.}. Any nonadaptive quantum algorithm computing a function $f(x)$ using $k$ queries to the input $x$ corresponds to a choice of a state $\ket{\psi}$, which is input to $k$ copies of the oracle (i.e.\ the unitary operator $O_x^{\otimes k}$), followed by a two-outcome measurement to determine whether $f(x)=0$ or $f(x)=1$.

One reason to study the nonadaptive model is that some important quantum algorithms are nonadaptive. An example is Simon's algorithm for the hidden subgroup problem over $\Z_2^n$~\cite{simon97}, which computes a {\em partial} function with bounded error and achieves an exponential speed-up over any possible classical algorithm. Separations are also known for total functions. For example, an optimal exact quantum algorithm for PARITY is nonadaptive and uses $\lceil n/2 \rceil$ quantum queries~\cite{cleve98a,farhi98a}. In the bounded-error setting, a more general result is known: by a remarkable result of van Dam, any boolean function of $n$ bits can be computed with bounded error using $n/2 + O(\sqrt{n})$ quantum queries~\cite{vandam98}. However, it has also been shown that this separation is close to optimal~\cite{montanaro10b}; any nonadaptive bounded-error quantum algorithm computing a total function depending on $n$ variables must make $\Omega(n)$ queries. Nonadaptive exact quantum query algorithms are even more restricted, and require at least $n/2$ queries~\cite{montanaro10b}.

Another motivation for studying the nonadaptive model is that, as we have seen, the general model of exact quantum query complexity appears to be rich and complex. Working in the much simpler nonadaptive model might allow stronger bounds and tighter characterisations to be proven.

For a boolean function $f$, let $D^{na}(f)$, $Q_E^{na}(f)$ be the nonadaptive quantum and classical exact query complexities of $f$, i.e.\ the minimum number of nonadaptive queries required to compute $f$ with certainty. The nonadaptive model is extremely restricted classically, as we see from the following easy proposition.

\begin{prop}
For any total boolean function $f$ depending on $n$ variables, $D^{na}(f)=n$.
\end{prop}

\begin{proof}
A nonadaptive exact classical query algorithm $\mathcal{A}$ making $k$ queries is specified by a list of $k$ fixed variables which are queried. If $k<n$, there must exist a variable $i$ which is not queried, but on which $f$ depends. Thus there must exist an input $x$ such that if bit $i$ is flipped, $\mathcal{A}$ does not notice the difference, so $\mathcal{A}$ cannot be correct on every input.
\end{proof}

We now introduce some additional notation. For any boolean function $f:\{0,1\}^n \rightarrow \{0,1\}$, define
\[ S_f := \{ z : \forall x, f(x) = f(x+z) \}, \]
where addition is over the group $\Z_2^n$; i.e.\ $S_f$ is the set of translations of the input under which $f$ is invariant. Note that $S_f$ is a subspace of $\{0,1\}^n$. For any subspace $S \subseteq \{0,1\}^n$, let $S^\perp$ denote the orthogonal subspace to $S$, i.e.\ $S^\perp = \{x:x \cdot s = 0, \forall\,s \in S\}$. Finally, let $d(x,S)$ denote the {\em maximum} Hamming distance between a bit-string $x \in \{0,1\}^n$ and a subset $S \subseteq \{0,1\}^n$: $d(x,S) = \max_{y \in S} d(x,y)$. Then we have the following theorem.

\begin{thm}
\label{thm:na}
For any boolean function $f:\{0,1\}^n \rightarrow \{0,1\}$,
\[ Q_E^{na}(f) = \min_{x \in \{0,1\}^n} \max_{y \in S_f^\perp} d(x,y) = \min_{x \in \{0,1\}^n} d(x,S_f^{\perp}). \]
\end{thm}

We have thus completely characterised the nonadaptive exact quantum query complexity of $f$. In the coding theory literature, the quantity $\min_{x \in \{0,1\}^n} d(x,S_f^{\perp})$ is known as the {\em radius} of the code $S_f^{\perp}$~\cite{cohen85}. Observe that Theorem \ref{thm:na} implies that $Q_E^{na}(f)$ can be computed exactly in time polynomial in $2^n$. We now prove this theorem and then draw some corollaries. In the proof it will be convenient to use the notation
\[ \hat{f}(x) = \frac{1}{2^n} \sum_{y \in \{0,1\}^n} (-1)^{x \cdot y} f(y) \]
for the Fourier transform (over $\Z_2^n$) of some function $f:\{0,1\}^n \rightarrow \R$.

\begin{proof}[Proof of Theorem \ref{thm:na}]
Label computational basis states by a length $k$ string of integers $(i_1,\dots,i_k)$ in the range $\{0,\dots,n\}$; each such string represents a list of variables queried, with 0 representing a ``null query'' which does nothing. $O_x^{\otimes k}$ acts on these basis states by mapping
\[ \ket{i_1,\dots,i_k} \mapsto (-1)^{x_{i_1} + \dots + x_{i_k}} \ket{i_1,\dots,i_k} . \]
Now note that we can restrict ourselves to query strings in non-decreasing order, and containing at most one of each integer between 1 and $n$. The first of these is because querying any permutation of a string is equivalent to querying the string itself. The second is because querying the same index twice does nothing, and hence is equivalent to the null query.

These strings are now in obvious one-to-one correspondence with the set of $n$-bit strings of Hamming weight at most $k$. Thus the state we obtain from applying $O_x^{\otimes k}$ to an arbitrary input state is of the form
\[ \ket{\psi_x} = \sum_{s\in\{0,1\}^n,|s|\le k} (-1)^{s \cdot x} \alpha_s \ket{s}. \]
A nonadaptive quantum query algorithm computing $f$ exactly using $k$ queries exists if and only if there exists a set $\{\alpha_s\}$ such that $\ip{\psi_x}{\psi_y}=0$ for all $x$, $y$ such that $f(x) \neq f(y)$, or in other words
\[ \sum_{s\in\{0,1\}^n,|s|\le k} (-1)^{s \cdot (x+y)} |\alpha_s|^2 = 0. \]
For brevity, write $w(s) = |\alpha_s|^2$. The above constraint says that, for all $z \notin S_f$,
\[ \sum_{s\in\{0,1\}^n,|s|\le k} (-1)^{s \cdot z} w(s) = 0. \]
Considering $w(s)$ as a function $w:\{0,1\}^n \rightarrow \R$ such that $w(s)=0$ for $|s|>k$, in Fourier-analytic terminology the constraint says that
\[ \widehat{w}(z) = 0 \mbox{ if } z \notin S_f, \]
i.e.\ that $ \widehat{w}$ is only supported on the subspace $S_f$. This is equivalent to the constraint that
\[ w(s) = \frac{1}{|S_f^\perp|} \sum_{t \in S_f^\perp} w(s+t) \mbox{ for all } s. \]
To see this, define the function $P_{S_f}(x) = [x \in S_f]$, and note that $\widehat{P_{S_f}}(t) = \frac{1}{|S_f^\perp|} [t \in S_{f}^\perp]$. Letting $\ast$ denote convolution over $\Z_2^n$ (i.e.\ $(f \ast g)(x) = \sum_y f(y) g(x+y)$), by Fourier duality we have
\beas
P_{S_f} w = w &\Leftrightarrow& \widehat{P_{S_f}} \ast \widehat{w} = \widehat{w} \Leftrightarrow \sum_t w(s+t) \widehat{P_{S_f}}(t) = w(s) \mbox{ for all } s\\
&\Leftrightarrow& \frac{1}{|S_f^\perp|} \sum_{t \in S_f^\perp} w(s+t) = w(s) \mbox{ for all } s.
\eeas
Thus $w$ is uniform on cosets of $S_f^\perp$. As $w$ is not identically zero, there must be a coset $t+S_f^\perp$ such that every element $s \in t + S_f^\perp$ has Hamming weight at most $k$. If $S_f = 0$, then $S_f^\perp = \{0,1\}^n$ and hence has only one coset, which contains $1^n$. Hence we must have $k=n$. More generally, we have that $Q_E^{na}(f)$ is the minimal $k$ such that there exists a $t$ satisfying $|s| \le k$ for all $s \in t + S_f^\perp$. In other words, $Q_E^{na}(f)$ is the minimal $k$ such that there exists a $t$ satisfying $d(s,t) \le k$ for all $s \in S_f^\perp$.
\end{proof}

We observe from this proof that it is without loss of generality that any nonadaptive exact quantum algorithm can be described as picking a coset of $S_f^\perp$ and querying everything in that subset uniformly. Explicitly, we have the following algorithm.

\begin{enumerate}
\item Let $t\in\{0,1\}^n$ be a bit-string such that $d(t,S_f^\perp) = k$.
\item Produce the state of $n$ qubits $\frac{1}{|S_f^\perp|^{1/2}} \sum_{s \in t + S_f^\perp} (-1)^{s \cdot x} \ket{s}$ at a cost of $k$ queries to the oracle.
\item Perform Hadamards on every qubit of the resulting state and measure to get outcome $\tilde{x}$.
\item Output $f(\tilde{x})$.
\end{enumerate}

One can easily verify that in fact $f(\tilde{x})=f(x)$ with certainty. We note that this is reminiscent of an algorithm of van Dam~\cite{vandam98} which learns $x$ itself with bounded error using $n/2 + O(\sqrt{n})$ queries to the oracle. Here, we also compute a partial Fourier transform from which $f(x)$ can be determined, but our algorithm succeeds with certainty.

We now draw some corollaries from Theorem \ref{thm:na}.

\begin{cor}
\label{cor:nan}
For any boolean function $f:\{0,1\}^n \rightarrow \{0,1\}$ such that $f$ is not invariant under any translation, $Q_E^{na}(f) = n$.
\end{cor}

\begin{proof}
If $f$ is not invariant under any translation, $S_f = \emptyset$ and hence $S_f^\perp = \{0,1\}^n$. For any bit-string $x \in \{0,1\}^n$, there exists a $y \in \{0,1\}^n$ such that $d(x,y)=n$. Hence $Q_E^{na}(f) = n$.
\end{proof}

One could also have observed this corollary by noting that $Q_E^{na}(f)$ depends only on $S_f$, and for the AND function (which has $Q_E(\text{AND})=n$~\cite{beals01}), $S_{\text{AND}} = \emptyset$. The corollary implies that only an exponentially small fraction of boolean functions $f$ have $Q_E^{na}(f) < n$. One class of functions that do satisfy this is given by the following corollary.

\begin{cor}
\label{cor:even}
For any boolean function $f:\{0,1\}^n \rightarrow \{0,1\}$ such that $f(x) = f(\bar{x})$ for all $x$, $Q_E^{na}(f) \le n-1$.
\end{cor}

\begin{proof}
The constraint $f(x) = f(\bar{x})$ for all $x$ is equivalent to $f(x+1^n) = f(x)$ for all $x$, so $\{0^n,1^n\} \subseteq S_f$, implying $S_f^\perp \subseteq \{ x: |x| \text{ even} \}$. If $n$ is odd, then $d(0^n,S_f^\perp) \le n-1$ (as $1^n$ has odd Hamming weight, there is no even weight bit string distance $n$ from $0^n$). Similarly, if $n$ is even, $d(10^{n-1},S_f^\perp) \le n-1$. Hence $Q_E^{na}(f) \le n-1$.
\end{proof}
An explicit nonadaptive quantum algorithm achieving this query complexity proceeds as follows. Evaluate $y_k := x_1 \oplus x_k$, for $2 \le k \le n$, at a cost of $n-1$ queries in total, then output $f(0,y_2,\dots,y_n)$. If $x_1=0$, this is simply $f(x)$, while if $x_1=1$, this is $f(\bar{x})=f(x)$.

\begin{cor}
\label{cor:nan2}
For any boolean function $f:\{0,1\}^n \rightarrow \{0,1\}$ that depends on all $n$ input bits, $Q_E^{na}(f) \ge \lceil n/2 \rceil$.
\end{cor}

\begin{proof}
We need to show that, for any bit-string $x$, we can find an element $y \in S_f^\perp$ such that $d(x,y) \ge n/2$. As $f$ depends on all its input bits, for all $i \in \{1,\dots,n\}$, $e_i \notin S_f$. This implies that, for all $i \in \{1,\dots,n\}$, there is at least one element of $S_f^\perp$ whose $i$'th bit is 1. 
Thus, if we pick an element $y \in S_f^\perp$ at random, each bit of $y$ will be 0 or 1 with equal probability, so the expectation of $d(x,y)$ is exactly $n/2$. Therefore, $d(x,S_f^\perp) \ge n/2$.
\end{proof}

Corollary \ref{cor:nan2} was previously proven in~\cite{montanaro10b} via a different method. We can also show that functions whose nonadaptive exact quantum query complexity is minimal are of very restricted form.

\begin{cor}
Let $f:\{0,1\}^n \rightarrow \{0,1\}$ be a boolean function that depends on all $n$ input bits and such that $Q_E^{na}(f)=n/2$. Then there exists an $x$ such that $d(x,y)=n/2$ for all $y \in S_f^\perp$.
\end{cor}

\begin{proof}
By the proof of Corollary \ref{cor:nan2}, $\E_{y \in S_f^\perp} d(x,y) = n/2$. Thus, if $d(x,y) \le n/2$ for {\em all} $y \in S_f^\perp$, we must have $d(x,y) = n/2$ for all $y \in S_f^\perp$.
\end{proof}


\subsection{Symmetric boolean functions}

It turns out that we can apply Theorem \ref{thm:na} to completely characterise the nonadaptive exact quantum query complexity of symmetric boolean functions, via the following tetrachotomy.

\begin{thm}
\label{thm:tetrachotomy}
Let $f:\{0,1\}^n \rightarrow \{0,1\}$ be symmetric. Then exactly one of the following four possibilities is true.
\begin{enumerate}
\item $f$ is constant and $Q_E^{na}(f)=0$.
\item $f$ is the PARITY function or its negation and $Q_E^{na}(f)=\lceil n/2 \rceil$.
\item $f$ satisfies $f(x) = f(\bar{x})$ (but is not constant, the PARITY function or its negation) and $Q_E^{na}(f) = n-1$.
\item $f$ is none of the above and $Q_E^{na}(f) = n$.
\end{enumerate}
\end{thm}

We will prove Theorem \ref{thm:tetrachotomy} using the following lemma, whose proof is given afterwards.

\begin{lem}
\label{lem:charpar}
Let $f:\{0,1\}^n \rightarrow \{0,1\}$ be symmetric and satisfy $f(x) = f(x+a)$ for all $x \in \{0,1\}^n$, for some $a$ with $1 \le |a| \le n-1$. Then, if $|a|$ is odd, $f$ is constant. If $|a|$ is even, $f$ is constant, PARITY or its negation.
\end{lem}

\begin{proof}[Proof of Theorem \ref{thm:tetrachotomy}]
First assume there is an $a$ with $1 \le |a| \le n-1$ such that $f(x) = f(x+a)$ for all $x \in \{0,1\}^n$. Then, by Lemma \ref{lem:charpar}, $f$ is constant, PARITY or its negation. If $f$ is constant then clearly $Q_E^{na}(f)=0$. If $f$ is PARITY or its negation, by the result~\cite{farhi98a} of Farhi et al.\ $Q_E^{na}(f)=\lceil n/2 \rceil$. On the other hand, if there is no $a$ with $1 \le |a| \le n-1$ such that $f(x) = f(x+a)$ for all $x \in \{0,1\}^n$, but there is such an $a$ with $|a|=n$, $f(x) = f(\bar{x})$ and by Corollary \ref{cor:even} $Q_E^{na}(f) = n-1$. Finally, if there is no $a \neq 0^n$ such that $f(x) = f(x+a)$ for all $x \in \{0,1\}^n$, by Corollary \ref{cor:nan} $Q_E^{na}(f)=n$.
\end{proof}

It will be convenient to prove Lemma \ref{lem:charpar} using Fourier analysis, based on the following well-known fact.

\begin{fact}
\label{fact:zeroinf}
For any $f:\{0,1\}^n \rightarrow \R$ and for any $a \in \{0,1\}^n$, if $f(x) = f(x+a)$ for all $x \in \{0,1\}^n$, then for all $b$ such that $|a \wedge b|$ is odd, $\hat{f}(b)=0$.
\end{fact}

\begin{proof}[Proof of Lemma \ref{lem:charpar}]
Note that, because $f$ is symmetric, if $\hat{f}(s)=0$ for some $s$ with $|s|=k$, $\hat{f}(t)=0$ for {\em all} $t$ with $|t|=k$. Without loss of generality, assume $a$ consists of $j$ ones followed by $n-j$ zeroes (i.e.\ is of the form $1\dots10\dots0$). Consider the bit-string $s$ which is 1 on the set $\{1,\dots,2k+1\}$ for some $0 \le k\le (j-1)/2$, and the set $\{n-\ell+1,\dots,\ell\}$, for some $1 \le \ell \le j$. By Fact \ref{fact:zeroinf}, for any such bit-string $s$, $\hat{f}(s)=0$. By varying $k$ and $\ell$, we can vary $|s|$ arbitrarily between 1 and either $n$ (if $|a|$ is odd), or $n-1$ (if $|a|$ is even). Thus, if $|a|$ is odd, $\hat{f}(s)=0$ for all $s \neq 0^n$, so $f$ is constant. Otherwise, if $|a|$ is even, $\hat{f}(s)=0$ for all $s \notin \{0^n,1^n\}$. The only boolean functions satisfying this are constant functions, PARITY and its negation.
\end{proof}


\section{Open problems}
\label{sec:open}

It is a very tempting conjecture that $Q_E(\text{EXACT}_k) = \max\{k,n-k\}$. It is easy to see that the lower bound $Q_E(\text{EXACT}_k) \ge \max\{k,n-k\}$ holds; by setting $\min\{k,n-k\}$ input bits to 0 we obtain a function equivalent to the AND function on $\ell := \max\{k,n-k\}$ bits, which has exact quantum query complexity $\ell$. So it would suffice to prove the upper bound $Q_E(\text{EXACT}_{n/2}) \le n/2$ for all even $n$ to prove this conjecture. To see this, note that for any $k \le n/2$, an algorithm for $\text{EXACT}_{k}$ on $n$ bits can be obtained from an algorithm for $\text{EXACT}_{n-k}$ on $2(n-k)$ bits simply by appending $n-2k$ bits set to 1; the case $k \ge n/2$ is similar.


Following the completion of this work, Ambainis has shown the existence of a total boolean function $f$ such that $Q_E(f) = O(D(f)^{0.8675\dots})$~\cite{ambainis12b}. It remains open to determine the optimal separation between quantum and classical exact query complexity; we are hopeful that the numerical techniques used in this paper may prove helpful in resolving this question, and in finding new examples of functions which demonstrate quantum-classical separations.


\section*{Acknowledgements}

AM was supported by an EPSRC Postdoctoral Research Fellowship and would like to thank Scott Aaronson and Dan Shepherd for comments. We would also like to thank two anonymous referees for their helpful suggestions. Special thanks to Andris Ambainis and Andrey Vihrov for pointing out an error in the algorithm of Section \ref{sec:balanced} in an earlier version of this paper.


\appendix

\section{Numerical results for functions on up to 6 bits}
\label{sec:numerical}

In this appendix we collate our numerical results concerning the optimal success probability achievable by quantum algorithms for all boolean functions on 4 input bits, and symmetric boolean functions on 5 and 6 input bits. We split the results into sections according to the number of bits on which the functions depend. Note that each section on functions of $k$ bits does not include functions which only depend on fewer than $k$ bits. In the following tables, entries are starred when there is a nonadaptive exact quantum algorithm  using that number of queries (see Section \ref{sec:nonadapt}). An entry ``1'' means that the SDP solver claims a solution with success probability greater than $0.999$; note that this does not strictly speaking imply the existence of an {\em exact} algorithm using that number of queries. We use the notation SYM$(c_0,\dots,c_n)$ to mean the symmetric function $f$ such that $f(x) = c_{|x|}$. In the tables of symmetric functions, we simply identify each function with the vector $(c_0,\dots,c_n)$. Note that for all non-constant symmetric $f$, the decision tree complexity $D(f) = n$, but this is not the case for $Q_E(f)$. For most functions $f$ on 4 bits, $D(f)$ is easily verified to be 4 via polynomial degree arguments; we calculated $D(f)$ for the remaining functions $f$ using the algorithm of~\cite{guijarro99}.

We also showed numerically that there exists no boolean function $f$ on up to 5 bits such that $Q_E(f) < D(f)/2$. There are too many functions on 5 bits to iterate through them na\"ively, so we used the following procedure. Any function $f$ on 5 bits such that $Q_E(f) \le 2$ can be obtained by setting $f(x) = (1-x_1)f_0(x_2,\dots,x_5) + x_1 f_1(x_2,\dots,x_5)$, where $f_0$ and $f_1$ are boolean functions on 4 bits such that $Q_E(f_0) \le 2$, $Q_E(f_1) \le 2$. From the previous numerical results there are 25 boolean functions $f$ on 4 bits, up to isomorphism, which have $Q_E(f) \le 2$. Generating and combining all functions isomorphic to these gives a large number of functions to test. However, the size of this list can be reduced by discounting all functions $f$ on 5 bits produced such that $\deg(f)=5$, and then running the efficient algorithm of~\cite{guijarro99} to discount all remaining functions $f$ such that $D(f) \le 4$. We are left with 13,608 candidate functions, none of which turn out to have $Q_E(f) \le 2$.


\subsection{Functions of 4 bits}

\begin{center}
\begin{tabular}{|c|c|c|c|c|c|}
\hline {\bf ID} & {\bf Function} & {\bf 1 query} & {\bf 2 queries} & {\bf 3 queries} & {\bf D(f)} \\
\hline 
1 & $x_1 \wedge x_2 \wedge x_3 \wedge x_4$ & 0.735 & 0.962 & 0.996 & 4\\ 
6 & & 0.654 & 0.931 & 1* & 4\\ 
7 & & 0.750 & 0.954 & 1 & 4\\ 
22 & & 0.572 & 0.906 & 1 & 4\\ 
23 & & 0.667 & 0.926 & 1 & 4\\ 
24 & $x_1 \wedge \neg \text{NAE}(\bar{x}_2,x_3,x_4)$ & 0.654 & 0.931 & 1* & 4 \\ 
25 & & 0.640 & 0.961 & 1 & 4\\ 
27 & $x_1 \wedge \text{SEL}(x_4,x_2,x_3)$ & 0.667 & 0.965 & 1 & 3\\ 
30 & & 0.600 & 0.956 & 1 & 4\\ 
31 & & 0.718 & 0.970 & 1 & 4\\ 
61 & & 0.643 & 0.976 & 1 & 4\\ 
105 & & 0.500 & 0.900 & 1* & 4\\ 
107 & & 0.571 & 0.941 & 1 & 4\\ 
111 & & 0.662 & 0.968 & 1* & 4\\ 
126 & $x_1 \wedge \text{NAE}(x_2,x_3,x_4)$ & 0.667 & 0.947 & 1* & 4\\ 
127 & & 0.727 & 0.972 & 1 & 4\\ 
278 & EXACT$_3$ & 0.529 & 0.884 & 1 & 4\\ 
279 & Th$_3$ & 0.643 & 0.900 & 1 & 4\\ 
280 & & 0.572 & 0.906 & 1 & 4\\ 
281 & & 0.600 & 0.956 & 1 & 4\\ 
282 & & 0.571 & 0.936 & 1 & 4\\ 
283 & & 0.637 & 0.959 & 1 & 4\\ 
286 & & 0.546 & 0.932 & 1 & 4\\ 
287 & & 0.659 & 0.945 & 1 & 4\\ 
300 & & 0.571 & 0.936 & 1 & 4\\ 
301 & & 0.572 & 0.964 & 1 & 4\\ 
303 & SEL$(x_3,x_1 \wedge x_2,\text{SEL}(x_4,x_1,x_2))$ & 0.644 & 0.966 & 1 & 3\\ 
316 & & 0.562 & 0.962 & 1 & 4\\ 
317 & SEL$(x_3,x_1 \wedge x_2,\text{SEL}(x_2,x_1,x_4))$ & 0.572 & 0.980 & 1 & 3\\ 
318 & & 0.546 & 0.956 & 1 & 4\\ 
319 & & 0.640 & 0.972 & 1 & 4\\ 
360 & & 0.529 & 0.884 & 1 & 4\\ 
361 & & 0.500 & 0.916 & 1 & 4\\
362 & & 0.546 & 0.932 & 1 & 4\\ 
363 & & 0.546 & 0.955 & 1 & 4\\ 
366 & & 0.546 & 0.956 & 1 & 4\\ 
367 & & 0.571 & 0.969 & 1 & 4\\ 
382 & & 0.546 & 0.923 & 1 & 4\\ 
383 & & 0.600 & 0.946 & 1 & 4\\ 
384 & NAE$(\bar{x}_1,x_2,x_3,x_4)$ & 0.800 & 0.980 & 1* & 4\\ 
385 & & 0.750 & 0.954 & 1 & 4\\ 
386 & & 0.640 & 0.961 & 1 & 4\\ 
\hline 
\end{tabular}
\end{center}

\begin{center}
\begin{tabular}{|c|c|c|c|c|c|}
\hline {\bf ID} & {\bf Function} & {\bf 1 q.} & {\bf 2 qs.} & {\bf 3 qs.} & {\bf D(f)}\\
\hline 
387 & & 0.667 & 0.965 & 1 & 4\\ 
390 & & 0.571 & 0.936 & 1 & 4\\ 
391 & & 0.637 & 0.959 & 1 & 4\\ 
393 & SEL$(x_3,x_1\wedge \bar{x}_4,x_2 \wedge x_4)$ & 0.667 & 0.965 & 1 & 3\\ 
395 & & 0.724 & 0.963 & 1 & 4\\ 
399 & & 0.751 & 0.980 & 1 & 4\\ 
406 & & 0.500 & 0.916 & 1 & 4\\ 
407 & & 0.572 & 0.940 & 1 & 4\\ 
408 & & 0.600 & 0.956 & 1 & 4\\ 
409 & & 0.643 & 0.976 & 1 & 4\\ 
410 & & 0.572 & 0.964 & 1 & 4\\ 
411 & & 0.656 & 0.969 & 1 & 4\\ 
414 & & 0.546 & 0.955 & 1 & 4\\ 
415 & & 0.642 & 0.965 & 1 & 4\\ 
424 & & 0.667 & 0.926 & 1 & 4\\ 
425 & & 0.637 & 0.959 & 1 & 4\\ 
426 & & 0.718 & 0.970 & 1 & 4\\ 
427 & SEL$(x_3,x_1 \wedge \bar{x}_4,\text{SEL}(x_4,x_1,x_2))$ & 0.751 & 0.980 & 1 & 3\\ 
428 & & 0.637 & 0.959 & 1 & 4\\ 
429 & SEL$(x_2,x_1 \wedge \bar{x}_4,\text{SEL}(x_3,x_1,x_4))$ & 0.656 & 0.969 & 1 & 3\\ 
430 & & 0.644 & 0.966 & 1 & 4\\ 
431 & & 0.710 & 0.977 & 1 & 4\\ 
444 & & 0.572 & 0.980 & 1 & 4\\ 
445 & & 0.641 & 0.965 & 1 & 4\\ 
446 & & 0.572 & 0.969 & 1 & 4\\ 
447 & & 0.667 & 0.980 & 1 & 4\\ 
488 & & 0.643 & 0.900 & 1 & 4\\ 
489 & & 0.572 & 0.940 & 1 & 4\\ 
490 & & 0.659 & 0.945 & 1 & 4\\ 
491 & & 0.642 & 0.965 & 1 & 4\\ 
494 & & 0.640 & 0.972 & 1 & 4\\ 
495 & SEL$(x_3,x_1,\text{SEL}(x_4,x_1,x_2))$ & 0.667 & 0.980 & 1 & 3\\ 
510 & & 0.600 & 0.946 & 1 & 4\\ 
829 & & 0.563 & 0.975 & 1 & 4\\ 
854 & & 0.598 & 0.955 & 1 & 4\\ 
855 & & 0.714 & 0.969 & 1 & 4\\ 
856 & & 0.572 & 0.964 & 1 & 4\\ 
857 & & 0.579 & 0.961 & 1 & 4\\ 
858 & SEL$(x_1,x_2 \wedge x_3,x2 \oplus x_4)$ & 0.572 & 0.980 & 1 & 3\\ 
859 & & 0.628 & 0.974 & 1 & 4\\ 
862 & & 0.572 & 0.966 & 1 & 4\\ 
863 & SEL$(x_1,x_2 \wedge x_3,x2 \vee x_4)$ & 0.667 & 0.986 & 1 & 3\\ 
872 & & 0.546 & 0.932 & 1 & 4\\ 
873 & & 0.500 & 0.946 & 1 & 4\\ 
874 & & 0.598 & 0.955 & 1 & 4\\ 
875 & & 0.572 & 0.951 & 1 & 4\\ 
876 & & 0.546 & 0.956 & 1 & 4\\ 
877 & & 0.545 & 0.961 & 1 & 4\\ 
\hline 
\end{tabular}
\end{center}

\begin{center}
\begin{tabular}{|c|c|c|c|c|c|}
\hline {\bf ID} & {\bf Function} & {\bf 1 q.} & {\bf 2 qs.} & {\bf 3 qs.} & {\bf D(f)}\\
\hline 
878 & & 0.572 & 0.966 & 1 & 4\\ 
879 & & 0.600 & 0.966 & 1 & 4\\ 
892 & & 0.563 & 0.975 & 1 & 4\\ 
893 & & 0.571 & 0.966 & 1 & 4\\ 
894 & & 0.572 & 0.947 & 1 & 4\\ 
961 & & 0.718 & 0.970 & 1 & 4\\ 
965 & SEL$(x_2,x_1\wedge\bar{x}_3,\text{SEL}(x_1,x_3,x_4))$ & 0.751 & 0.980 & 1 & 3\\ 
966 & & 0.644 & 0.966 & 1 & 4\\ 
967 & & 0.710 & 0.977 & 1 & 4\\ 
980 & & 0.659 & 0.945 & 1 & 4\\ 
981 & & 0.714 & 0.969 & 1 & 4\\ 
982 & & 0.572 & 0.951 & 1 & 4\\ 
983 & & 0.661 & 0.965 & 1 & 4\\ 
984 & SEL$(x_1,x_2\wedge x_3,\text{SEL}(x_4,\bar{x}_3,\bar{x}_2))$ & 0.644 & 0.966 & 1 & 3\\ 
985 & & 0.628 & 0.974 & 1 & 4 \\ 
987 & SEL$(x_4,\text{SEL}(x_3,x_1,x_2),\text{SEL}(x_2,x_1,x_3))$ & 0.661 & 0.965 & 1 & 3\\ 
988 & & 0.640 & 0.972 & 1 & 4\\ 
989 & SEL$(x_1,x_2 \wedge x_3,\bar{x}_3 \vee x_4)$ & 0.667 & 0.986 & 1 & 3\\ 
990 & & 0.600 & 0.966 & 1 & 4\\ 
1632 & $(x_1 \oplus x_2) \wedge (x_3 \oplus x_4)$ & 0.667 & 1* & 1* & 4\\ 
1633 & & 0.562 & 0.962 & 1 & 4\\ 
1634 & & 0.643 & 0.976 & 1 & 4\\ 
1635 & & 0.572 & 0.980 & 1 & 4\\ 
1638 & & 0.667 & 0.947 & 1* & 4\\ 
1639 & & 0.641 & 0.965 & 1 & 4\\ 
1641 & & 0.500 & 0.936 & 1* & 4\\ 
1643 & & 0.561 & 0.966 & 1 & 4\\ 
1647 & MAJ$(x_1,x_2,x_3 \oplus x_4)$ & 0.667 & 1 & 1* & 4\\ 
1650 & & 0.656 & 0.969 & 1 & 4\\ 
1651 & & 0.628 & 0.974 & 1 & 4\\ 
1654 & & 0.641 & 0.965 & 1 & 4\\ 
1656 & & 0.546 & 0.956 & 1 & 4\\ 
1657 & & 0.500 & 0.964 & 1 & 4\\ 
1658 & & 0.571 & 0.966 & 1 & 4\\ 
1659 & & 0.571 & 0.962 & 1 & 4\\ 
1662 & & 0.600 & 0.954 & 1 & 4\\ 
1680 & $(x_1 \oplus x_2) \wedge (x_1 \oplus x_3 \oplus x_4)$ & 0.500 & 0.900 & 1* & 4\\ 
1681 & & 0.500 & 0.916 & 1 & 4\\ 
1683 & & 0.500 & 0.946 & 1 & 4\\ 
1686 & & 0.500 & 0.936 & 1* & 4\\ 
1687 & & 0.500 & 0.964 & 1 & 4\\ 
1695 & SEL$(x_3 \oplus x_4,x_1,x_2)$ & 0.500 & 1 & 1* & 3\\ 
1712 & & 0.571 & 0.941 & 1 & 4\\ 
1713 & & 0.546 & 0.955 & 1 & 4\\ 
1714 & & 0.572 & 0.940 & 1 & 4\\ 
1715 & & 0.572 & 0.951 & 1 & 4\\ 
1716 & & 0.546 & 0.955 & 1 & 4\\ 
1717 & & 0.545 & 0.961 & 1 & 4\\ 
\hline 
\end{tabular}
\end{center}

\begin{center}
\begin{tabular}{|c|c|c|c|c|c|}
\hline {\bf ID} & {\bf Function} & {\bf 1 q.} & {\bf 2 qs.} & {\bf 3 qs.} & {\bf D(f)}\\
\hline 
1718 & & 0.561 & 0.966 & 1 & 4\\ 
1719 & SEL$(x_4,\text{SEL}(x_2,x_1,x_3),\text{SEL}(x_3,x_2,x_1))$ & 0.572 & 0.962 & 1 & 3\\ 
1721 & & 0.500 & 0.964 & 1 & 4\\ 
1725 & & 0.529 & 0.955 & 1 & 4\\ 
1776 & & 0.662 & 0.967 & 1* & 4\\ 
1777 & & 0.572 & 0.969 & 1 & 4\\ 
1778 & & 0.642 & 0.965 & 1 & 4\\ 
1782 & SEL$(x_2,x_1,x_3 \oplus x_4)$ & 0.667 & 1 & 1* & 3 \\ 
1785 & $x_1 \oplus (x_2 \wedge (x_3 \oplus x_4))$ & 0.500 & 1 & 1* & 4\\ 
1910 & & 0.600 & 0.954 & 1 & 4\\ 
1912 & & 0.546 & 0.923 & 1 & 4\\ 
1913 & & 0.529 & 0.955 & 1 & 4\\ 
1914 & & 0.572 & 0.947 & 1 & 4\\ 
1918 & & 0.572 & 0.922 & 1 & 4\\ 
1968 & & 0.662 & 0.968 & 1* & 4\\ 
1969 & & 0.642 & 0.965 & 1 & 4\\ 
1972 & & 0.572 & 0.969 & 1 & 4\\ 
1973 & & 0.600 & 0.966 & 1 & 4\\ 
1974 & & 0.572 & 0.962 & 1 & 4\\ 
1980 & SEL$(x_3,\text{SEL}(x_4,x_1,x_2),x_1 \oplus x_2)$ & 0.571 & 0.966 & 1 & 3\\ 
2016 & SEL$(x_1,x_2 \wedge(x_3 \vee x_4), \bar{x}_2 \wedge (\bar{x}_3 \vee \bar{x}_4))$ & 0.773 & 1 & 1* & 4\\ 
2017 & & 0.640 & 0.972 & 1 & 4\\ 
2018 & & 0.710 & 0.977 & 1 & 4\\ 
2019 & & 0.667 & 0.986 & 1 & 4\\ 
2022 & SEL$(x_3,\text{SEL}(x_2,x_1,x_4),\text{SEL}(x_1,x_2,\bar{x}_4))$ & 0.661 & 0.965 & 1 & 3\\ 
2025 & & 0.571 & 0.966 & 1 & 4\\ 
2032 & & 0.727 & 0.971 & 1 & 4\\ 
2033 & & 0.667 & 0.980 & 1 & 4\\ 
2034 & SEL$(x_2,x_1,\text{SEL}(x_4,x_3,\bar{x}_1))$ & 0.667 & 0.980 & 1 & 3\\ 
2040 & & 0.600 & 0.946 & 1 & 4\\ 
5736 & EXACT$_2$ & 0.572 & 1 & 1* & 4\\ 
5737 & SYM(0,0,1,0,1) & 0.500 & 0.962 & 1 & 4\\ 
5738 & & 0.563 & 0.975 & 1 & 4\\ 
5739 & & 0.500 & 0.980 & 1 & 4\\ 
5742 & & 0.572 & 0.947 & 1 & 4\\ 
5758 & & 0.572 & 0.922 & 1 & 4\\ 
5761 & & 0.500 & 0.860 & 1 & 4\\ 
5763 & & 0.500 & 0.907 & 1 & 4\\ 
5766 & & 0.500 & 0.936 & 1* & 4\\ 
5767 & & 0.500 & 0.933 & 1 & 4\\ 
5769 & & 0.500 & 0.907 & 1 & 4\\ 
5771 & & 0.500 & 0.946 & 1 & 4\\ 
5774 & & 0.561 & 0.966 & 1 & 4\\ 
5782 & & 0.500 & 0.962 & 1 & 4\\ 
5783 & & 0.500 & 0.954 & 1 & 4\\ 
5784 & & 0.500 & 0.946 & 1 & 4\\ 
5785 & & 0.500 & 0.933 & 1 & 4\\ 
5786 & & 0.500 & 0.964 & 1 & 4\\ 
\hline 
\end{tabular}
\end{center}

\begin{center}
\begin{tabular}{|c|c|c|c|c|c|}
\hline {\bf ID} & {\bf Function} & {\bf 1 q.} & {\bf 2 qs.} & {\bf 3 qs.} & {\bf D(f)}\\
\hline 
5787 & & 0.500 & 0.955 & 1 & 4\\ 
5790 & SEL$(x_3,\text{SEL}(x_4,x_1,x_2),x_2 \oplus x_4)$ & 0.500 & 0.980 & 1 & 3\\ 
5801 & & 0.500 & 0.933 & 1 & 4\\ 
5804 & & 0.545 & 0.961 & 1 & 4\\ 
5805 & & 0.500 & 0.955 & 1 & 4\\ 
5820 & & 0.529 & 0.955 & 1 & 4\\ 
5865 & & 0.500 & 0.954 & 1 & 4\\ 
6014 & SYM(0,0,1,1,0) & 0.600 & 0.874 & 1 & 4\\ 
6030 & SEL$(x_3,\text{SEL}(x_4,x_1,x_2),\text{SEL}(x_4,x_2,\bar{x}_1))$ & 0.667 & 1 & 1* & 3 \\ 
6038 & & 0.500 & 0.980 & 1 & 4\\ 
6040 & & 0.572 & 0.951 & 1 & 4\\ 
6042 & SEL$(x_4,\text{SEL}(x_3,x_1,x_2),\text{SEL}(x_2,x_3,\bar{x}_1))$ & 0.572 & 0.962 & 1 & 3\\ 
6060 & & 0.600 & 0.966 & 1 & 4\\ 
6120 & $x_1 \oplus \text{MAJ}(x_2,x_3,x_4)$ & 0.667 & 1 & 1* & 4\\ 
6375 & $x_1 \oplus \neg \text{NAE}(\bar{x}_2,x_3,x_4)$ & 0.500 & 0.900 & 1* & 4\\ 
6625 & & 0.500 & 0.946 & 1 & 4\\ 
6627 & & 0.500 & 0.955 & 1 & 4\\ 
6630 & & 0.500 & 0.954 & 1 & 4\\ 
7128 & Sorted input bits~\cite{ambainis06} & 0.854 & 1 & 1* & 3\\ 
7140 & $x_1 \oplus \text{SEL}(x_4,x_2,x_3)$ & 0.500 & 1 & 1* & 3\\ 
7905 & & 0.500 & 0.900 & 1* & 4\\ 
27030 & \text{PARITY} & 0.500 & 1* & 1* & 4\\ 
\hline 
\end{tabular}
\end{center}

\vfill

\subsection{Symmetric functions of 5 bits}

\begin{center}
\begin{tabular}{|c|c|c|c|c|}
\hline {\bf Function} & {\bf 1 query} & {\bf 2 queries} & {\bf 3 queries} & {\bf 4 queries}\\
\hline 
(0,0,0,0,0,1) & 0.693 & 0.925 & 0.988 & 0.999 \\
(0,0,0,0,1,0) & 0.516 & 0.761 & 0.972 & 1 \\
(0,0,0,0,1,1) & 0.640 & 0.798 & 0.974 & 1 \\
(0,0,0,1,0,0) & 0.530 & 0.616 & 1 & 1 \\
(0,0,0,1,0,1) & 0.500 & 0.593 & 0.995 & 1 \\
(0,0,0,1,1,0) & 0.546 & 0.758 & 1 & 1 \\
(0,0,0,1,1,1) & 0.600 & 0.728 & 1 & 1 \\
(0,0,1,0,0,1) & 0.500 & 0.640 & 0.988 & 1 \\
(0,0,1,0,1,0) & 0.500 & 0.517 & 1 & 1 \\
(0,0,1,0,1,1) & 0.500 & 0.534 & 1 & 1 \\
(0,0,1,1,0,0) & 0.600 & 0.874 & 1 & 1* \\
(0,0,1,1,0,1) & 0.500 & 0.856 & 0.998 & 1 \\
(0,0,1,1,1,0) & 0.616 & 0.762 & 0.969 & 1 \\
(0,1,0,0,0,1) & 0.500 & 0.728 & 0.967 & 1 \\
(0,1,0,0,1,0) & 0.500 & 0.860 & 1 & 1* \\
(0,1,0,1,0,1) & 0.500 & 0.500 & 1* & 1* \\
(0,1,0,1,1,0) & 0.500 & 0.616 & 0.998 & 1 \\
(0,1,1,0,0,1) & 0.500 & 0.784 & 0.998 & 1 \\
(0,1,1,1,1,0) & 0.736 & 0.962 & 0.996 & 1* \\
\hline
\end{tabular}
\end{center}

\subsection{Symmetric functions of 6 bits}

\begin{center}
\begin{tabular}{|c|c|c|c|c|c|}
\hline {\bf Function} & {\bf 1 query} & {\bf 2 queries} & {\bf 3 queries} & {\bf 4 queries} & {\bf 5 queries}\\
\hline
(0,0,0,0,0,0,1) & 0.663 & 0.900 & 0.980 & 0.997 & 0.9999\\
(0,0,0,0,0,1,0) & 0.511 & 0.684 & 0.940 & 0.993 & 1\\
(0,0,0,0,0,1,1) & 0.640 & 0.738 & 0.946 & 0.993 & 1\\
(0,0,0,0,1,0,0) & 0.516 & 0.572 & 0.878 & 1 & 1\\
(0,0,0,0,1,0,1) & 0.500 & 0.541 & 0.875 & 0.999 & 1\\
(0,0,0,0,1,1,0) & 0.527 & 0.751 & 0.904 & 1 & 1\\
(0,0,0,0,1,1,1) & 0.589 & 0.710 & 0.901 & 1 & 1\\
(0,0,0,1,0,0,0) & 0.530 & 0.616 & 1 & 1 & 1*\\
(0,0,0,1,0,0,1) & 0.500 & 0.614 & 0.980 & 0.997 & 1\\
(0,0,0,1,0,1,0) & 0.500 & 0.504 & 0.946 & 1 & 1\\
(0,0,0,1,0,1,1) & 0.500 & 0.525 & 0.952 & 1 & 1\\
(0,0,0,1,1,0,0) & 0.546 & 0.667 & 0.864 & 1 & 1\\
(0,0,0,1,1,0,1) & 0.500 & 0.625 & 0.860 & 1 & 1\\
(0,0,0,1,1,1,0) & 0.556 & 0.721 & 0.905 & 1 & 1\\
(0,0,1,0,0,0,1) & 0.500 & 0.583 & 0.882 & 0.997 & 1\\
(0,0,1,0,0,1,0) & 0.500 & 0.527 & 0.839 & 1 & 1 \\
(0,0,1,0,0,1,0) & 0.500 & 0.541 & 0.843 & 1 & 1 \\
(0,0,1,0,1,0,0) & 0.500 & 0.517 & 1 & 1 & 1* \\
(0,0,1,0,1,0,1) & 0.500 & 0.500 & 0.985 & 1 & 1 \\
(0,0,1,0,1,1,0) & 0.500 & 0.520 & 0.940 & 1 & 1 \\
(0,0,1,1,0,0,1) & 0.500 & 0.712 & 0.867 & 1 & 1 \\
(0,0,1,1,0,1,0) & 0.500 & 0.513 & 0.840 & 1 & 1 \\
(0,0,1,1,1,0,0) & 0.616 & 0.762 & 0.969 & 1 & 1* \\
(0,0,1,1,1,0,1) & 0.500 & 0.722 & 0.965 & 1 & 1 \\
(0,0,1,1,1,1,0) & 0.625 & 0.702 & 0.939 & 0.992 & 1 \\
(0,1,0,0,0,0,1) & 0.500 & 0.652 & 0.934 & 0.992 & 1 \\
(0,1,0,0,0,1,0) & 0.500 & 0.728 & 0.967 & 1 & 1* \\
(0,1,0,0,1,0,1) & 0.500 & 0.500 & 0.836 & 1 & 1 \\
(0,1,0,0,1,1,0) & 0.500 & 0.667 & 0.863 & 1 & 1 \\
(0,1,0,1,0,0,1) & 0.500 & 0.500 & 0.934 & 1 & 1 \\
(0,1,0,1,0,1,0) & 0.500 & 0.500 & 1* & 1* & 1* \\
(0,1,0,1,1,1,0) & 0.500 & 0.553 & 0.880 & 0.999 & 1 \\
(0,1,1,0,0,0,1) & 0.500 & 0.758 & 0.908 & 1 & 1 \\
(0,1,1,0,1,1,0) & 0.500 & 0.640 & 0.988 & 1 & 1* \\
(0,1,1,1,1,1,0) & 0.693 & 0.925 & 0.988 & 0.999 & 1* \\
\hline
\end{tabular}
\end{center}


\pagebreak

\section{Source code}
\label{sec:cvxcode}

The following is an example of how the CVX package~\cite{cvx} can be used to determine quantum query complexity. For full source code, see~\cite{websiteqc}. In this case, we calculate the minimal error probability over all quantum algorithms using 2 queries to compute some function $f:\{0,1\}^3 \rightarrow \{0,1\}$ (given as a column vector).

\begin{verbatim}
cvx_begin

  cvx_precision best;

  % variables m_i^j : 0 <= i <= n, 0 <= j <= t-1
  variable m00(8,8) symmetric; variable m10(8,8) symmetric;
  variable m20(8,8) symmetric; variable m30(8,8) symmetric;
  variable m01(8,8) symmetric; variable m11(8,8) symmetric;
  variable m21(8,8) symmetric; variable m31(8,8) symmetric;
    
  variable g0(8,8) symmetric; variable g1(8,8) symmetric;
  
  variable epss;
  
  minimise( epss );
  
  subject to
  
    % Input condition.
    m00 + m10 + m20 + m30 == ones(8,8);
    
    % Running conditions (between 1 and t-1).
    m01 + m11 + m21 + m31 == E0 .* m00 + E1 .* m10 + E2 .* m20 + E3 .* m30;
    
    % Output matches last but one query.
    g0 + g1 == E0 .* m01 + E1 .* m11 + E2 .* m21 + E3 .* m31;
    
    % Output constraints.
    diag(g0) >= (1-epss)*(1-f);
    diag(g1) >= (1-epss)*f;
    
    % Semidefinite constraints.
    m00 == semidefinite(8); m10 == semidefinite(8);
    m20 == semidefinite(8); m30 == semidefinite(8);
    m01 == semidefinite(8); m11 == semidefinite(8);
    m21 == semidefinite(8); m31 == semidefinite(8);
        
    g0 == semidefinite(8); g1 == semidefinite(8);

cvx_end
\end{verbatim}


\bibliographystyle{plain}
\bibliography{../thesis}

\end{document}